\def\version{2 May 2015}

\documentclass[reqno]{amsart}
\usepackage{amssymb,graphics,graphicx,bbm,hyperref}

\def\be{\begin{equation}}
\def\ee{\end{equation}}
\def\bm{\begin{multline}}
\def\bfig{\begin{figure}[htb]}
\def\efig{\end{figure}}
\newcommand{\dd}{{\rm d}}
\newcommand{\e}[1]{\,{\rm e}^{#1}\,}
\newcommand{\ii}{{\rm i}}
\def\tr{{\operatorname{tr\,}}}
\def\Tr{{\operatorname{Tr\,}}}
\def\supp{{\operatorname{supp\,}}}

\newcommand{\sumtwo}[2]{\sum_{\substack{#1 \\ #2}}}
\newcommand{\sumthree}[3]{\sum_{\substack{#1 \\ #2 \\ #3}}}

\newcommand{\suptwo}[2]{\sup_{\substack{#1 \\ #2}}}
\newcommand{\mintwo}[2]{\min_{\substack{#1 \\ #2}}}
\newcommand{\nnorm}[1]{|\!|\!| #1 |\!|\!|}

\numberwithin{equation}{section}
\newtheorem{theorem}{Theorem}[section]

\newtheorem{lemma}[theorem]{Lemma}
\newtheorem{corollary}[theorem]{Corollary}

\newcommand{\caA}{{\mathcal A}}
\newcommand{\caC}{{\mathcal C}}

\newcommand{\caH}{{\mathcal H}}
\newcommand{\caL}{{\mathcal L}}
\newcommand{\caM}{{\mathcal M}}
\newcommand{\caS}{{\mathcal S}}
\newcommand{\bbC}{{\mathbb C}}
\newcommand{\bbN}{{\mathbb N}}

\newcommand{\bbR}{{\mathbb R}}

\newcommand{\bbZ}{{\mathbb Z}}

\def\bbone{{\mathchoice {\rm 1\mskip-4mu l} {\rm 1\mskip-4mu l} {\rm 1\mskip-4.5mu l} {\rm 1\mskip-5mu l}}}

\newcommand{\Id}{\mathrm{\texttt{Id}}}

\makeatletter
  \def\tagform@#1{\maketag@@@{\footnotesize{(#1)}\@@italiccorr}}
\makeatother

\renewcommand{\eqref}[1]{(\ref{#1})}

\begin{document}

{\hfill\small \version} \vspace{2mm}

\title[Some Properties of Correlations of Quantum Lattice Systems]{Some Properties of Correlations of Quantum Lattice Systems in Thermal Equilibrium}

\author{J\"urg Fr\"ohlich}
\address{Institut f\"ur Theoretische Physik, ETH Z\"urich, Switzerland}
\email{juerg@phys.ethz.ch}

\author{Daniel Ueltschi}
\address{Department of Mathematics, University of Warwick,
Coventry, CV4 7AL, United Kingdom}
\email{daniel@ueltschi.org}



\begin{abstract}
Simple proofs of uniqueness of the thermodynamic limit of KMS states and of the decay of equilibrium correlations are presented for a large class of quantum lattice systems at high temperatures. New quantum correlation inequalities for general Heisenberg models are described. Finally, a simplified derivation of a general result on power-law decay of correlations in 2D quantum lattice systems with continuous symmetries is given, extending results of Mc Bryan and Spencer for the 2D classical XY model.\\

\qquad \textit{We dedicate this note to the memory of our friend Oscar E. Lanford III.}
\end{abstract}

\thanks{\copyright{} 2015 by the authors. This paper may be reproduced, in its
entirety, for non-commercial purposes.}

\maketitle


\section{Introduction}

Quantum lattice systems have been widely studied for many decades, heuristically, numerically and mathematically. Many important rigorous results on equilibrium phase transitions and broken symmetries have been discovered for such systems at \textit{low enough} temperatures. Surveys of such results can be found, e.g., in \cite{BR, DLS, FILS, BKUel, DFF}, and references given there.
 
In this note, we study a general class of quantum lattice systems (see Sect. 2) in thermal equilibrium and present simple proofs of two basic results valid at \textit{high enough} temperatures: (i) the uniqueness of the KMS state in the thermodynamic limit; and (ii) exponential decay of correlations. We also establish: (iii) power-law decay of equilibrium correlations at \textit{arbitrary} temperatures in two-dimensional quantum lattice systems with continuous symmetries. Variants of all these results have been described in the literature; see \cite{BR, MS, KP, Uel} and references given there. Our purpose, in this note, is to delineate a natural level of generality for these results and to present simple or simplified proofs thereof. Furthermore, we derive some new correlation inequalities for quantum spin systems in thermal equilibrium. These inequalities do not appear to be as useful as, e.g., the GKS- and FKG inequalities known to hold for certain classes of classical lattice systems; yet, they contain useful information on the dependence of correlations on some coupling constants. In essence, our inequalities say that correlations among spin components become stronger if the coupling constants of the interaction terms among these spin components in the Hamiltonian are increased.

\section{Uniqueness of KMS state at high temperatures}

It is well-known that, at sufficiently high temperatures, there are no phase-transitions, and one expects that equilibrium states are unique. This claim is backed by various mathematical results, such as analyticity of the free energy at high temperatures. In this section, we show that, for a large class of quantum lattice systems, assuming that the temperature is high enough, only a single state satisfies the KMS condition that characterizes thermal equilibrium in quantum systems. We refer the reader to the monograph of Bratteli and Robinson \cite{BR} for a survey of earlier such results and references to the literature. These authors remark, in particular, that O. E. Lanford III observed that a uniqueness theorem follows from an earlier result due to Greenberg. Here we propose to present a variant of Lanford's approach and an improved estimate on the critical temperature. We think that our proof is somewhat simpler than the arguments described in \cite{BR}. The basic idea involved in all proofs we are aware of, including ours, is to use the KMS condition to derive an inhomogeneous linear equation for the correlators of an equilibrium state satisfying the KMS condition and to show that, at high enough temperatures, this equation has a unique solution (under suitable assumptions on the interactions specifying the particular quantum lattice system; see also \cite{Dir}).

For concreteness, we study quantum lattice systems on the simple (hyper) cubic lattice $\bbZ^{d}$.
Let $\caH_{x} = \bbC^{N}$ denote the Hilbert space of pure state vectors of the quantum-mechanical degrees of freedom, e.g., a quantum-mechanical spin, located at the site $x \in \bbZ^{d}$, and let 
$\caA_{x} = \caM_{N}(\bbC)$ denote the algebra of bounded linear operators acting on $\caH_{x}$, with $N<\infty$ \textit{independent} of $x \in \mathbb{Z}^{d}$. For an arbitrary finite subset 
$\Lambda \subset \bbZ^{d}$, we define
\be
\caH_{\Lambda} = \bigotimes_{x\in\Lambda} \caH_{x},
\ee
and we let $\caA_{\Lambda} = \otimes_{x\in\Lambda} \caA_{x}$ denote the algebra of bounded operators on $\caH_{\Lambda}$. If $\Lambda \subset \Lambda'$, we view $\caA_{\Lambda}$ as a subalgebra of 
$\caA_{\Lambda'}$ by identifying $A \in \caA_{\Lambda}$ with $A \otimes \bbone_{\Lambda'\setminus\Lambda} \in \caA_{\Lambda'}$.

Let $(\Phi_{X})_{X \subset \bbZ^{d}}$ denote an ``interaction'', that is, a collection of operators $\Phi_{X} \in \caA_{X}$, for any finite subset $X$ of $\bbZ^{d}$. The norm of an interaction is defined by
\be
\|\Phi\|_{r} = \sup_{x\in\bbZ^{d}} \sum_{X \ni x} \|\Phi_{X}\| r^{|X|}.
\ee
Here, $\|\Phi_{X}\|$ denotes the usual operator norm in $\caA_{X}$, and $r\geq1$ is a parameter. The Hamiltonian associated with a finite domain $\Lambda \subset \bbZ^{d}$ is given by
\be
\label{def Ham}
H_{\Lambda} = \sum_{X \subset \Lambda} \Phi_{X}.
\ee
For $t \in \bbC$, let $\alpha_{t}^{\Lambda}$ be the linear automorphism of $\caA_{\Lambda}$ that describes the time evolution of operators (``observables'') in $\caA_{\Lambda}$, namely
\be
\alpha_{t}^{\Lambda}(A) := \e{\ii t H_{\Lambda}} A \e{-\ii t H_{\Lambda}}.
\ee
In order to describe infinite systems, we consider the $C^{*}$-algebra, $\mathcal{A}$, of quasi-local observables, which is the norm-completion of the usual algebra of local observables
\be
\caA = \overline{\mathcal{A}_{0}}, \qquad \text{where} \qquad
\mathcal{A}_{0} := \bigvee_{\Lambda \nearrow \bbZ^{d}} \caA_{\Lambda}.
\ee
It is well-known that if $\|\Phi\|_{r} < \infty$, for some $r>1$, there exists a unique one-parameter group of $^{*}$automorphisms of $\caA$, $\alpha_{t} : \caA \to \caA$, with $t\in \mathbb{R}$, such that
\be
\lim_{n\to\infty} \| \alpha_{t}^{\Lambda_{n}}(A) - \alpha_{t}(A) \| = 0,
\ee
for an arbitrary local observable $A$ and any sequence of domains $(\Lambda_{n})$ increasing to 
$\bbZ^{d}$; that is, such that any finite set $\Lambda$ is contained in all $\Lambda_{n}$'s, as soon as $n$ is large enough (depending on $\Lambda$). The operator function $\alpha_{t}(A)$ has an analytic continuation  in $t$ to the complex plane, for all $A \in \mathcal{A}_0$. A ``state'' is a bounded, positive, normalized linear functional on 
$\caA$. A state $\rho$ describes \textit{thermal equilibrium at inverse temperature} $\beta$ iff it satisfies the KMS condition, i.e., iff
\be
\label{KMS}
\rho \bigl( AB \bigr) = \rho \bigl( B \, \alpha_{\ii\beta}(A) \bigr),
\ee
for all $A,B$ in $\caA_{0}$. By considering sequences of finite-volume (Gibbs) equilibrium states, a standard compactness argument shows the existence of cluster points of states that satisfy the KMS condition, i.e., the existence of KMS states is an almost trivial fact. We are now prepared to state our uniqueness theorem.

\begin{theorem}
\label{thm uniqueness}
Assume that
\[
\beta \, \|\Phi\|_{N+1} < (2N)^{-1}.
\]
Then there exists a unique KMS state at inverse temperature $\beta$.
\end{theorem}

We actually prove the theorem under the more general condition that there exists $s<1/N$ such that $2\beta \|\Phi\|_{N (1+s)} < s$.
As mentioned above, the strategy of our proof is to reformulate the KMS condition as an equation for the equilibrium state that has a unique solution when $\beta$ is small enough. In order to derive this equation, we express observables as commutators of operators. The proof of Theorem \ref{thm uniqueness} will be given after the one of Lemma \ref{lem decomp}, which we state next. 

Here and in the sequel, $\|\cdot\|_{\rm HS}$ denotes the normalized Hilbert-Schmidt norm
\be
\| A \|_{\rm HS}^{2} = \frac1{\dim \caH_{\Lambda}} \Tr A^{*} A.
\ee
Notice that
\be
\frac1{\sqrt{\dim \caH_{\Lambda}}} \|A\| \leq \|A\|_{\rm HS} \leq \|A\|
\ee
for all $A \in \caA_{\Lambda}$.

\begin{lemma}
\label{lem decomp}
Let $A$ be a hermitian $N \times N$ matrix with the property that $\Tr A = 0$. Then there exist hermitian 
$N \times N$ matrices $B_{1},\dots, B_{N-1}$ and $C_{1},\dots,C_{N-1}$ such that
\[
\begin{split}
&A = \sum_{i=1}^{N-1} [B_{i}, C_{i}], \\
&\sum_{i=1}^{N-1} \|B_{i}\|_{\rm HS} \, \|C_{i}\|_{\rm HS} \leq \sqrt N \, \|A\|_{\rm HS}.
\end{split}
\]
\end{lemma}

\begin{proof}
Let $a_{1},\dots,a_{N}$ be the eigenvalues of $A$ (repeated according to their multiplicity). We have that
\be
\sum_{i=1}^{N} a_{i} = 0, \qquad \sum_{i=1}^{N} |a_{i}|^{2} = N \, \|A\|_{\rm HS}^{2}.
\ee
In particular, each $|a_{i}|$ is bounded above by $\sqrt N \|A\|_{\rm HS}$. Let us order the eigenvalues so that
\be
\Bigl| \sum_{i=1}^{k} a_{i} \Bigr| \leq \sqrt N \, \|A\|_{\rm HS}
\ee
for all $1 \leq k \leq N-1$. This is indeed possible, as can be seen by induction using $\sum a_{i} = 0$: If $0 \leq \sum^{k} a_{i} \leq \sqrt N \|A\|_{\rm HS}$, we can find $a_{k+1} \leq 0$ among the remaining eigenvalues such that $| \sum^{k+1} a_{i} | \leq \sqrt N \|A\|_{\rm HS}$. And if the partial sum is negative, we can find $a_{k+1} \geq 0$ among the remaining eigenvalues, with the same conclusion.

We work in a basis such that $A$ is diagonal and its eigenvalues are ordered so they satisfy the properties above.
Let $\tilde a_{k} = \sum_{i=1}^{k} a_{i}$, and let $\sigma_{j,j+1}^{1}, \sigma_{j,j+1}^{2}, \sigma_{j,j+1}^{3}$ be $N \times N$ matrices that are equal to Pauli matrices on the $2\times2$ block that contains $(j,j)$ and $(j+1,j+1)$, and that are equal to zero everywhere else. It is not hard to check that
\be
A = \sum_{j=1}^{N-1} \tilde a_{j} \, \sigma_{j,j+1}^{3}.
\ee
We therefore have that
\be
A = \tfrac12 \sum_{j=1}^{N-1} \tilde a_{j} \, [\sigma_{j,j+1}^{1}, \sigma_{j,j+1}^{2}],
\ee
which proves the first claim. The bound follows from $|\tilde a_{j}| \leq \sqrt N \|A\|_{\rm HS}$ and $\|\sigma_{j,j+1}^{i}\|_{\rm HS}^{2} = 2/N$.
\end{proof}

\begin{proof}[Proof of Theorem \ref{thm uniqueness}]
Let $(e_{i})_{i=0}^{N^{2}-1}$ be a hermitian basis of $\caM_{N}(\bbC)$, with $e_{0}=\bbone$, $\Tr e_{i} = 0$ if $1 \neq 0$, and $\| e_{i} \| = 1$, for all $i$. Let $J$ be the set of multi-indices $j = (j_{x})_{x \in \bbZ^{d}}$, $0 \leq j_{x} \leq N^{2}-1$, with finite support
\be
\supp j = \{ x \in \bbZ^{d} | j_{x} \neq 0 \}.
\ee
Given $j \in J$, let $e_{j} = \otimes_{x \in \supp j} e_{j_{x}} \in \caA_{\supp j}$. The linear span of $\{ e_{j} \}_{j \in J}$ is dense in $\caA$.

Let $\tr$ denote the normalized trace on $\caA$; it is equal to $\frac1{\dim \caH_{\Lambda}} \Tr$ on $\caA_{\Lambda}$ and it can be extended to $\caA$ by continuity. The state $\rho$ can be written as $\rho = \tr + \varepsilon$ where $\varepsilon(\bbone)=0$. We actually have that
\be
\varepsilon(e_{j}) = \begin{cases} \rho(e_{j}) & \text{if } j \not\equiv 0, \\ 0 & \text{if } j\equiv0. \end{cases}
\ee
Using Lemma \ref{lem decomp}, we have that
\be
e_{j} = \frac1{|\supp j|} \sum_{y \in \supp j} \sum_{i=1}^{N-1} \bigl[ \otimes_{x \neq y} e_{j_{x}} \otimes b_{i}^{(j_{y})}, \otimes_{x \neq y} \bbone \otimes c_{i}^{(j_{y})} \bigr],
\ee
for $j \not\equiv 0$. Here, $b_{i}^{(k)}, c_{i}^{(k)}$ are the matrices $B_{i}, C_{i}$ of Lemma \ref{lem decomp} in the case where the matrix $A$ is $e_{k}$.

We now use this decomposition and the KMS condition \eqref{KMS} in order to get an equation for $\varepsilon$. For $j \not\equiv 0$,
\be
\label{eq from KMS}
\begin{split}
\varepsilon(e_{j}) &= \frac1{|\supp j|} \sum_{y \in \supp j} \sum_{i=1}^{N-1} \rho \Bigl( \bigl[ \otimes_{x \neq y} e_{j_{x}} \otimes b_{i}^{(j_{y})}, \otimes_{x \neq y} \bbone \otimes c_{i}^{(j_{y})} \bigr] \Bigr) \\
&= \frac1{|\supp j|} \sum_{y \in \supp j} \sum_{i=1}^{N-1} \rho \Bigl( \otimes_{x \neq y} e_{j_{x}} \otimes b_{i}^{(j_{y})} \cdot (\bbone - \alpha_{\ii\beta}) \otimes_{x \neq y} \bbone \otimes c_{i}^{(j_{y})} \Bigr) \\
&= \delta(e_{j}) + K_{\beta} \varepsilon(e_{j}).
\end{split}
\ee
In the above equation, we set
\be
\delta(e_{j}) = \frac1{|\supp j|} \sum_{y \in \supp j} \sum_{i=1}^{N-1} \tr \Bigl( \otimes_{x \neq y} e_{j_{x}} \otimes b_{i}^{(j_{y})} \cdot (\bbone - \alpha_{\ii\beta}) \otimes_{x \neq y} \bbone \otimes c_{i}^{(j_{y})} \Bigr),
\ee
and the operator $K_{\beta}$ is defined by
\be
\label{def K beta}
(K_{\beta} \phi)(e_{j}) = \frac1{|\supp j|} \sum_{y \in \supp j} \sum_{i=1}^{N-1} \phi \Bigl( \otimes_{x \neq y} e_{j_{x}} \otimes b_{i}^{(j_{y})} \cdot (\bbone - \alpha_{\ii\beta}) \otimes_{x \neq y} \bbone \otimes c_{i}^{(j_{y})} \Bigr).
\ee
Notice that $K_{\beta}$ is a linear operator on the Banach space $\caL(\caA)$ of linear functionals on $\caA$. Equation \eqref{eq from KMS} can be written as
\be
\label{the smart equation}
(\bbone - K_{\beta}) \varepsilon = \delta.
\ee

Let us introduce the following norm on $\caL(\caA)$:
\be
\nnorm{\phi} = \sup_{j \in J} |\phi(e_{j})|.
\ee
Because $\| e_{j} \| = 1$ for all $j$, we have $\nnorm{\phi} \leq \|\phi\|$ and $(\caL(\caA), \nnorm\cdot)$ is a normed vector space. We consider $K_{\beta}$ as an operator on $(\caL(\caA), \nnorm\cdot)$ and we show that its norm is strictly less than 1; the solution of \eqref{the smart equation} is then unique. The norm of $K_{\beta}$ is equal to
\be
\|K_{\beta}\| = \sup_{\nnorm\phi = 1} \sup_{j \in J} |K_{\beta} \phi(e_{j})|.
\ee

Recall that $\alpha_{\ii\beta} = \lim_{\Lambda} \alpha_{\ii\beta}^{\Lambda}$ (with convergence in the operator norm) and that $\alpha_{\ii\beta}^{\Lambda}(A)$, $A \in \caA$, has a well-known expansion in multiple commutators. From \eqref{def K beta}, we get
\bm
\bigl| K_{\beta} \phi(e_{j}) \bigl| \leq \frac1{|\supp j|} \sum_{y \in \supp j} \sum_{i=1}^{N-1} \sum_{n\geq1} \frac{\beta^{n}}{n!} \sup_{\Lambda \subset \bbZ^{d}} \sum_{X_{1}, \dots, X_{n} \subset \Lambda} \\
\Bigl| \phi \Bigl( \otimes_{x \neq y} e_{j_{x}} \otimes b_{i}^{(j_{y})} \bigl[ \Phi_{X_{n}}, \dots, \bigl[ \Phi_{X_{1}}, \otimes_{x \neq y} \bbone \otimes c_{i}^{(j_{y})} \bigr] \cdots \bigr] \Bigr) \Bigr|.
\end{multline}
Because of the commutators, the sum over the $X_{k}$'s is restricted to subsets that satisfy
\be
\label{set constraints}
\begin{split}
&X_{1} \ni y, \\
&X_{2} \cap X_{1} \neq \emptyset, \\
&\qquad \vdots \\
&X_{n} \cap (X_{1} \cup \dots \cup X_{n-1}) \neq \emptyset.
\end{split}
\ee
Let $A = \sum_{(j_{x}')_{x\in X}} a_{j'} e_{j'}$ be an operator in $\caA_{X}$. For any $(j_{x})_{x \notin X}$, we have
\be
\label{sonder bound}
\begin{split}
\bigl| \phi \bigl( \otimes_{x \notin X} e_{j_{x}} \otimes A \bigr) \bigr| &= \Bigl| \sum_{(j_{x}')_{x \in X}} a_{j'} \; \phi \bigl( \otimes_{x \notin X} e_{j_{x}} \otimes_{x \in X} e_{j_{x}'} \bigr) \Bigr| \\
&\leq \nnorm\phi \sum_{(j_{x}')_{x \in X}} |a_{j'}| \\
&\leq \nnorm\phi \|A\|_{\rm HS} N^{|X|}.
\end{split}
\ee
Using Eq.\ \eqref{sonder bound} with $\nnorm\phi = 1$, $\|AB\|_{\rm HS} \leq \|A\| \, \|B\|_{\rm HS}$, and $\|c_{i}^{(j_{y})}\| \leq \sqrt N \|c_{i}^{(j_{y})}\|_{\rm HS}$, we get
\be
\begin{split}
\bigl| K_{\beta} \phi(e_{j}) \bigl| &\leq \sqrt N \sup_{y \in \bbZ^{d}} \sum_{n\geq1} \frac{(2\beta)^{n}}{n!} \sum_{X_{1}, \dots, X_{n} : y} \Bigl( \prod_{k=1}^{n} \|\Phi_{X_{k}}\| N^{|X_{k}|} \Bigr) \sum_{i=1}^{N-1} \|b_{i}^{(j_{y})}\|_{\rm HS} \|c_{i}^{(j_{y})}\|_{\rm HS} \\
&\leq N \sup_{y \in \bbZ^{d}} \sum_{n\geq1} \frac{(2\beta)^{n}}{n!} \sum_{X_{1},\dots,X_{n} : y} \; \prod_{k=1}^{n} \|\Phi_{X_{k}}\| N^{|X_{k}|}.
\end{split}
\ee
We have used Lemma \ref{lem decomp} to get the last line.
The constraint $X_{1}, \dots, X_{n} : y$ means that \eqref{set constraints} must be respected.
The final step is to estimate the sum over such subsets. This can be conveniently done with an inductive argument. Namely, let $R_{0}=0$ and, for $m\geq1$, let
\be
R_{m} = \sup_{y \in \bbZ^{d}} \sum_{n=1}^{m} \frac{(2\beta)^{n}}{n!} \sum_{X_{1},\dots,X_{n} : y} \prod_{k=1}^{n} \|\Phi_{X_{k}}\| N^{|X_{k}|}.
\ee
Summing first over $X_{1} \ni y$, then over sets that intersect sites of $X_{1}$, we get
\be
\begin{split}
R_{m} &\leq 2\beta \sup_{y} \sum_{X_{1} \ni y} \|\Phi_{X_{1}}\| N^{|X_{1}|} \prod_{x \in X_{1}} \Bigl( \sum_{n=1}^{m} \frac{(2\beta)^{n-1}}{(n-1)!} \sum_{X_{2},\dots,X_{n} : x} \prod_{k=2}^{n} \|\Phi_{X_{k}}\| N^{|X_{k}|} \Bigr) \\
&\leq 2\beta \sup_{y} \sum_{X_{1} \ni y} \|\Phi_{X_{1}}\| N^{|X_{1}|} (1+R_{m-1})^{|X_{1}|}.
\end{split}
\ee
It follows easily that $R_{m} \leq r$ for all $m$, and all $r$ such that $2\beta \|\Phi\|_{N (1+r)} \leq r$. Then $\|K_{\beta}\| \leq Nr$, and the assumption of Theorem \ref{thm uniqueness} implies the existence of $r$ such that $Nr < 1$.
\end{proof}

\section{High temperature expansions}
\label{sec high}

(Connected) correlations between operators localized in disjoint regions of the lattice vanish when $\beta=0$. For positive, but small $\beta$ and short-range interactions, correlations decay exponentially fast. This can be proven in several different ways. Here we use the method of \textit{cluster expansions}, which is robust and applies to both classical and quantum systems. The main result of this section and our method of proof are not new; see \cite[Section V.5]{Sim} and references therein. Our approach is based on the simple exposition in \cite{Uel}. It is quite direct and straightforward.

As an alternative to cluster expansions, one should mention the method of Lee and Yang, i.e., general \textit{Lee-Yang theorems}. This method establishes and then exploits analyticity properties of correlation functions in variables corresponding to external magnetic fields. It yields exponential decay of correlations, provided the magnetic field variables belong to certain subsets of the complex plane. We do not wish to describe these matters in more detail here; but see \cite{LP,GRS,Pfi,FR} for precise statements of results and proofs.

\subsection{Analyticity of the free energy}

Let $\Lambda$ be a finite subset of $\bbZ^{d}$.
Let $\caS_{\Lambda}$ denote the set of finite sequences $(X_{1},\dots,X_{n})$, with $n\geq1$ and $X_{i} \subset \Lambda$ for all $i$. Let $\caC_{\Lambda} \subset \caS_{\Lambda}$ denote the set of {\it clusters}, i.e., the set of objects $C = (X_{1},\dots,X_{n})$ such that the graph
\be
\bigl\{ (i,j) : X_{i} \cap X_{j} \neq \emptyset \bigr\}
\ee
is connected. We also let $\supp C = \cup_{i} X_{i}$ denote the support of $C$. We introduce the following weight function on $\caS_{\Lambda}$: If $C = (X_{1},\dots,X_{n})$,
\be
w(C) = \frac{\beta^{n}}{n!} \tr \Phi_{X_{1}} \dots \Phi_{X_{n}}.
\ee
Finally, let $\varphi$ denote the the usual combinatorial function of cluster expansions, namely
\be
\label{def phi}
\varphi(C_{1},\dots,C_k) = \begin{cases} 1 & \text{if } k = 1, \\ \sum_{g \in {\rm Conn}(k)} \prod_{\{i,j\} \in g} (-1_{\supp C_{i} \cap \supp C_{j} \neq \emptyset}) & \text{if } k\geq 2. \end{cases}
\ee
Here, ${\rm Conn}(k)$ is the set of connected graphs of $k$ vertices, and the product is over the edges of the connected graph $g$.

The first result deals with the partition function
\be
Z_{\Lambda} = \tr \e{-\beta H_{\Lambda}},
\ee
with $H_{\Lambda}$ the Hamiltonian defined in Eq.\ \eqref{def Ham}. As before, tr denotes the normalized trace. It follows easily from Theorem \ref{thm part fct} that the free energy $f_{\Lambda}(\beta) = -\frac1{|\Lambda|} \log Z_{\Lambda}$ is analytic in $\beta$ in the infinite-volume limit.

\begin{theorem}
\label{thm part fct}
Assume that there exists $a>0$ such that
\[
\beta \|\Phi\|_{\e{a} (1+a)} \leq a.
\]
Then the partition function has the expression
\[
Z_{\Lambda} = \exp \Bigl\{ \sum_{k\geq1} \frac1{k!} \sum_{C_{1}, \dots. C_{k} \in \caC_{\Lambda}} \varphi(C_{1},\dots,C_{k}) \prod_{i=1}^{k} w(C_{i}) \Bigr\}.
\]
The sums are absolutely convergent, and
\[
1 + \sum_{k\geq2} k \sum_{C_{1},\dots,C_{k} \in \caC_{\Lambda}} \bigl| \varphi(C_{2},\dots,C_{k}) \bigr| \prod_{i=2}^{k} |w(C_{i})| \leq \e{a |\supp C_{1}|},
\]
for all $\Lambda \subset \bbZ^{d}$ and all $C_{1} \in \caC_{\Lambda}$.
\end{theorem}

We remark that, historically, the ``clusters'' of the expansion are the connected sets of ${\rm Conn}(k)$ in Eq.\ \eqref{def phi} rather than our $C_{i}$s. Clusters are often grouped according to their supports, which yields the ``polymer'' expansion. But we find it better to keep the $C_{i}$s as they are, without resummation.

\begin{proof}
Clearly,
\be
\label{expansion Z}
Z_{\Lambda} = \tr \e{-\beta \sum_{X \subset \Lambda} \Phi_{X}} = \sum_{n\geq0} \frac{\beta^{n}}{n!} \sum_{X_{1},\dots,X_{n} \subset \Lambda} \tr \Phi_{X_{1}} \dots \Phi_{X_{n}}.
\ee
We group the sets $X_{1}, \dots, X_{n}$ in clusters. We get
\be
Z_{\Lambda} = \sum_{k\geq0} \frac1{k!} \sumtwo{C_{1},\dots,C_{k} \in \caC_{\Lambda}}{\rm disjoint} w(C_{1}) \dots w(C_{k}).
\ee
The sum is restricted on ``disjoint'' clusters such that $\supp C_{i} \cap \supp C_{j} = \emptyset$ for all $i \neq j$. This expression fits the framework of the method of cluster expansion. A sufficient condition for its convergence \cite{GK,KP,Uel} is that there exists $a>0$ such that
\be
\label{KP crit}
\sumtwo{C' \in \caC_{\Lambda}}{\supp C' \cap \supp C \neq \emptyset} |w(C')| \e{a |\supp C'|} \leq a |\supp C|,
\ee
for all $C \in \caC_{\Lambda}$. Once \eqref{KP crit} is proved, Theorem \ref{thm part fct} follows immediately from e.g.\ \cite[Theorem 1]{Uel}.

Let $n(C)$ denote the number of sets that constitute the cluster $C$. We have
\be
|w(C)| \leq \frac{\beta^{n(C)}}{n(C)!} \prod_{i=1}^{n(C)} \| \Phi_{X_{i}} \|.
\ee
Let $R_{0}=0$, and, for $m\geq1$,
\be
R_{m} = \sup_{x \in \bbZ^{d}} \sumtwo{C \in \caC_{\Lambda}, \supp C \ni x}{n(C) \leq m} \frac{\beta^{n(C)}}{n(C)!} \prod_{i=1}^{n(C)} \| \Phi_{X_{i}} \| \e{a |\supp C|}.
\ee
We show that $R_{m} \leq a$ for all $m$ (and all $\Lambda$); this implies \eqref{KP crit}. We prove it by induction by means of the inequality \eqref{induction ineq} below. We now give a careful derivation.

Let $x \in \Lambda$, and let us consider an order on the subsets of $\Lambda$ with the property that $X \prec X'$ if $X \ni x \notin X'$. If $f$ is a nonnegative function on subsets of $\Lambda$, and writing $f(C) = \prod f(X_{i})$, we have
\be
\begin{split}
\frac1{n!} \sumtwo{C = (X_{1}, \dots, X_{n}) \in \caC_{\Lambda}}{\supp C \ni x} &\prod_{i=1}^{n} f(X_{i}) \leq \sumtwo{X_{1} \prec \dots \prec X_{n}}{X_{1} \ni x, (X_{1},\dots,X_{n}) \in \caC_{\Lambda}} \prod_{i=1}^{n} f(X_{i}) \\
&= \sum_{X_{1} \ni x} f(X_{1}) \sum_{k\geq0} \frac1{k!} \sumthree{C_{1}, \dots, C_{k} \in \caC_{\Lambda}, {\rm disjoint}}{n(C_{1}) + \dots + n(C_{k}) = n-1}{X_{1} \prec C_{i}, \supp C_{i} \cap X_{1} \neq \emptyset \, \forall i} \prod_{i=1}^{k} \frac1{n(C_{i})!} f(C_{i}).
\end{split}
\ee
The inequality in the first line is due to the case of identical sets, $X_{i}=X_{j}$ for some $i \neq j$.
In the last sum, the constraint $X_{1} \prec C_{i}$ means that $X_{1}$ is smaller than all the sets of $C_{i}$. It follows that
\be
\label{induction ineq}
\begin{split}
R_{m} &\leq \beta \sup_{x\in\bbZ^{d}} \sum_{X \ni x} \| \Phi_{X} \| \e{a|X|} \prod_{y\in X} \biggl( 1 + \sumtwo{C \in \caC_{\Lambda}, \supp C \ni y}{n(C) \leq m-1} \frac{\beta^{n(C)}}{n(C)!} \prod_{i=1}^{n(C)} \| \Phi_{X_{i}} \| \e{a |\supp C|} \biggr) \\
&\leq \beta \sup_{x \in \bbZ^{d}} \sum_{X \ni x} \| \Phi_{X} \| \bigl( \e{a} (1 + R_{m-1}) \bigr)^{|X|}.
\end{split}
\ee
Using the induction hypothesis $R_{m-1} \leq a$ and the assumption of the theorem, we get $R_{m} \leq a$. This proves \eqref{KP crit}.
\end{proof}

\subsection{Thermodynamic limit and expectations of local observables}

Next, we consider the expectation of observables. Let $A \in \caA_{\Lambda}$. We let $\supp A$ denote the support of the observable $A$; it is equal to the smallest set $X$ such that $A \in \caA_{X}$. We are interested in the expectation
\be
\langle A \rangle = \frac1{Z_{\Lambda}} \tr A \e{-\beta H_{\Lambda}}.
\ee
A similar expansion than \eqref{expansion Z} gives
\be
\label{mean expansion}
\tr A \e{-\beta H_{\Lambda}} = \sum_{k\geq0} \frac1{k!} \sumtwo{C_{A}, C_{1}, \dots, C_{k}}{\rm disjoint} w_{A}(C_{A}) w(C_{1}) \dots w(C_{k}),
\ee
where
\be
w_{A}(C_{A}) = \frac{\beta^{n}}{n!} \tr A \Phi_{X_{1}} \dots \Phi_{X_{n}}.
\ee
Here, $C_{A} = (X_{0}, X_{1}, \dots, X_{n})$ is a cluster such that $X_{0} = \supp A$ by definition.
$n=0$ is possible in which case $w_{A}(C_{A}) = \tr A$.

Under the same assumption as in Theorem \ref{thm part fct}, the method of cluster expansion applies and it gives
\be
\tr A \e{-\beta H_{\Lambda}} = \sum_{C_{A}} w_{A}(C_{A}) \exp\biggl\{ \sum_{k\geq1} \frac1{k!} \sumtwo{C_{1},\dots, C_{k} \in \caC_{\Lambda}}{C_{A}, C_{1}, \dots, C_{k} \, {\rm disjoint}} \varphi(C_{1},\dots,C_{k}) \prod_{i=1}^{k} w(C_{i}) \biggr\}.
\ee
This can be combined with the expression for $Z_{\Lambda}$ in Theorem \ref{thm part fct}, because of cancellations, we obtain
\be
\langle A \rangle = \sum_{C_{A}} w_{A}(C_{A}) \exp\biggl\{ \sum_{k\geq1} \frac1{k!} \sumtwo{C_{1},\dots,C_{k} \in \caC_{\Lambda}}{(\cup_{i} \supp C_{i}) \cap \supp C_{A} \neq \emptyset} \varphi(C_{1},\dots,C_{k}) \prod_{i=1}^{k} w(C_{i}) \biggr\}.
\ee
This expression makes it possible to take the thermodynamic limit $\Lambda \nearrow \bbZ^{d}$ as all sums converge absolutely and uniformly.

\subsection{Decay of two-point correlation functions}

Let $b(X)$ be a nonnegative function of finite subsets of $\bbZ^{d}$. The larger this function, the better the decay. We assume a slightly stronger condition on the interaction, namely that there exists $a>0$ such that
\be
\label{stronger condition}
\beta \sup_{x \in \bbZ^{d}} \sum_{X \ni x} \| \Phi_{X} \| \e{\frac32 a|X| + b(X)} < a.
\ee

Given two sets $X, Y \subset \bbZ^{d}$, let
\be
\mu(X,Y) = \min_{n\geq1} \mintwo{X_{1},\dots,X_{n}}{X_{i} \cap X_{i+1} \neq \emptyset \; \forall i=0,\dots,n+1} b(X_{1}) + \dots + b(X_{n}).
\ee
In the second minimum, we set $X_{0}=X$ and $X_{n+1}=Y$.

\begin{theorem}
\label{thm decay corr}
Assume that the interaction $\Phi$ satisfies the condition \eqref{stronger condition}. Then we have
\[
\bigl| \langle AB \rangle - \langle A \rangle \langle B \rangle \bigr| \leq k(A,B) \, \e{-\mu(\supp A, \supp B)}
\]
with
\[
k(A,B) = \|A\| \, \|B\| \, \bigl( a |\supp A| + a |\supp B| + 3a^{2} |\supp A| \, |\supp B| \bigr).
\]
\end{theorem}

As $A$ and $B$ are moved away from each other, the decay is given by $\e{-\mu(\cdot)}$. Decay is exponential if the interactions are finite-range or exponentially decaying.

\begin{proof}
An expansion similar to \eqref{mean expansion} holds in the case where $A$ is replaced by the product of two operators, $AB$. We denote $C_{AB}$ the clusters of the type $(\supp A, \supp B, X_{1},\dots, X_{n})$; $n=0$ is not possible unless $\supp A \cap \supp B \neq \emptyset$. The corresponding weight is
\be
w_{AB}(C_{AB}) = \frac{\beta^{n}}{n!} \tr AB \Phi_{X_{1}} \dots \Phi_{X_{n}}.
\ee
Expansion of the exponential gives
\be
\begin{split}
\tr AB \e{-\beta H_{\Lambda}} &= \sum_{k\geq0} \frac1{k!} \sumtwo{C_{AB}, C_{1}, \dots, C_{k}}{\rm disjoint} w_{AB}(C_{AB}) w(C_{1}) \dots w(C_{k}) \\
&+ \sum_{k\geq0} \frac1{k!} \sumtwo{C_{A}, C_{B}, C_{1}, \dots, C_{k}}{\rm disjoint} \frac1{k!} w_{A}(C_{A}) w_{B}(C_{B}) w(C_{1}) \dots w(C_{k}).
\end{split}
\ee

It is convenient to use the following notation, which mirrors that of \cite[Section 3]{Uel}.
\be
\begin{split}
Z_{\Lambda}(C_{A}) &= \sum_{k\geq0} \frac1{k!} \sumtwo{C_{1},\dots,C_{k} \in \caC_{\Lambda \setminus \supp C_{A}}}{\rm disjoint} w(C_{1}) \dots w(C_{k}), \\
Z_{\Lambda}(C_{A}, C_{B}) &= \sum_{k\geq0} \frac1{k!} \sumtwo{C_{1},\dots,C_{k} \in \caC_{\Lambda \setminus (\supp C_{A} \cup \supp C_{B})}}{\rm disjoint} w(C_{1}) \dots w(C_{k}), \\
\hat Z_{\Lambda}(C_{A}) &= \sum_{k\geq0} (k+1) \sum_{C_{1},\dots,C_{k} \in \caC_{\Lambda}} \varphi(C_{A}, C_{1}, \dots, C_{k}) w(C_{1}) \dots w(C_{k}), \\
\hat Z_{\Lambda}(C_{A}, C_{B}) &= \sum_{k\geq0} (k+1) (k+2) \sum_{C_{1},\dots,C_{k} \in \caC_{\Lambda}} \varphi(C_{A}, C_{B}, C_{1}, \dots, C_{k}) w(C_{1}) \dots w(C_{k}).
\end{split}
\ee
We have
\be
\tr AB \e{-\beta H_{\Lambda}} = \sum_{C_{AB}} w_{AB}(C_{AB}) Z_{\Lambda}(C_{AB}) + \sumtwo{C_{A},C_{B}}{\rm disjoint} w_{A}(C_{A}) w_{B}(C_{B}) Z_{\Lambda}(C_{A},C_{B}).
\ee
It follows from \cite[Theorem 2]{Uel} that
\be
\label{tr corr}
\begin{split}
\langle AB \rangle - \langle A \rangle \langle B \rangle = &\sum_{C_{AB}} w_{AB}(C_{AB}) \hat Z_{\Lambda}(C_{AB}) + \sumtwo{C_{A}, C_{B}}{\rm disjoint} w_{A}(C_{A}) w_{B}(C_{B}) \hat Z_{\Lambda}(C_{A}, C_{B}) \\
&- \sumtwo{C_{A}, C_{B}}{\supp C_{A} \cap \supp C_{B} \neq \emptyset} w_{A}(C_{A}) w_{B}(C_{B}) \hat Z_{\Lambda}(C_{A}) \hat Z_{\Lambda}(C_{B}).
\end{split}
\ee

Let $b(C) = \sum_{i} b(X_{i})$ for $C = (X_{1},\dots,X_{n})$. Adapting the proof of \eqref{KP crit}, one can show that
\be
\label{KP crit bis}
\sumtwo{C' \in \caC_{\Lambda}}{\supp C' \cap \supp C \neq \emptyset} |w(C')| \e{\frac32 a |\supp C'| + b(C')} \leq a |\supp C|.
\ee
This allows to use \cite[Theorem 3]{Uel}. For $C_{A} = (\supp A, X_{1},\dots,X_{n})$, we get
\be
\e{b(C_{A})} |\hat Z_{\Lambda}(C_{A})| \leq \e{a |\supp C_{A}| + b(\supp A)}.
\ee
The same bound applies to $C_{B}$; as for $C_{AB}$, we have
\be
\label{last bound}
\e{\mu(\supp C_{A}, \supp C_{B})} |\hat Z_{\Lambda}(C_{A},C_{B})| \leq \e{\frac32 a |\supp C_{A}| + \frac32 a |\supp C_{B}|}.
\ee
Theorem \ref{thm decay corr} follows from the expression \eqref{tr corr} and the bounds \eqref{KP crit bis}--\eqref{last bound}.
\end{proof}

\section{Correlation inequalities for quantum spin systems}

We now consider a more restricted setting.
Let $S^{1}, S^{2}, S^{3}$ be spin operators in $\bbC^{N}$ that satisfy $[S^{1},S^{2}] = \ii S^{3}$, and with the other commutation relations obtained by cyclic permutations of indices. Let $S_{x}^{i} = S^{i} \otimes \Id_{\Lambda\setminus\{x\}}$, $i=1,2,3$. The Hamiltonian depends on real coupling parameters (exchange couplings), $J_{xy}^{i} = J_{yx}^{i}$, and is given by
\be
\label{def spin Ham}
H_{\Lambda} = -\tfrac12 \sum_{x,y \in \Lambda} \Bigl( J^{1}_{xy} S_{x}^{1} S_{y}^{1} + J_{xy}^{2} S_{x}^{2} S_{y}^{2} + J_{xy}^{3} S_{x}^{3} S_{y}^{3} \Bigr).
\ee
Here, $\Lambda$ is an arbitrary finite set.
From now on, we use the usual trace, denoted $\Tr$, rather than the normalized trace, $\tr$. With $Z_{\Lambda} = \Tr \e{-\beta H_{\Lambda}}$ denoting the partition function, the correlation functions at inverse temperature $\beta$ are given by
\be
\langle S_{0}^{i} S_{x}^{i} \rangle = \frac1{Z_{\Lambda}} \Tr (S_{0}^{i} S_{x}^{i} \e{-\beta H_{\Lambda}}).
\ee

The case $J_{xy}^{1} = J_{xy}^{2} = 0$, for all $x,y\in\Lambda$, corresponds to the Ising model, which is in fact a classical model. The case $J_{xy}^{3}=0$ and $J_{xy}^{1} = J_{xy}^{2}$, for all $x,y$, corresponds to the quantum XY model. And the symmetric case, $J_{xy}^{1} = J_{xy}^{2} = J_{xy}^{3}$, corresponds to the isotropic Heisenberg model. Positive values of the couplings correspond to ferromagnetic order, while negative values of the couplings correspond to antiferromagnetism.

It is natural to expect that correlations are stronger among those components of the spins that correspond to stronger coupling parameters in the Hamiltonian. This is the content of the next theorem. The inequalities stated there do not appear to have been noticed before, except for the spin-$\frac{1}{2}$ XY model corresponding to $N=2$: Assuming that $\Lambda$ is a rectangular subset of $\bbZ^{d}$ and $\|x\|=1$, the first inequality follows from reflection positivity \cite{Kubo}; for general $\Lambda$ and general $x$, it follows from a random loop representation \cite{Uel2}.

\begin{theorem}
\label{thm quantum correl ineq}
Assume that, for all $x,y \in \Lambda$, the couplings satisfy
\[
|J_{xy}^{2}| \leq J_{xy}^{1}.
\]
Then we have that
\[
\bigl| \langle S_{0}^{2} S_{x}^{2} \rangle \bigr| \leq \langle S_{0}^{1} S_{x}^{1} \rangle,
\]
for all $x \in \Lambda$.
More generally, for all $x_{1},\dots x_{k} \in \Lambda$ and $j_{1},\dots,j_{k} \in \{1,2\}$,
\[
\bigl| \langle S_{x_{1}}^{j_{1}} \dots S_{x_{k}}^{j_{k}} \rangle \bigr| \leq \langle S_{x_{1}}^{1} \dots S_{x_{k}}^{1} \rangle.
\]
\end{theorem}

Further inequalities can be generated using symmetries. Some inequalities hold for the staggered two-point function $(-1)^{|x|} \langle S_{0}^{i} S_{x}^{i} \rangle$.

\begin{proof}
Let $S \in \frac12 \bbN$ such that $2S+1 = N$, and let $|a\rangle$, $a \in \{-S,\dots,S\}$ denote basis elements of $\bbC^{2S+1}$. Let the operators $S^{\pm}$ be defined by
\be
S^{+}|a\rangle = \sqrt{S(S+1) - a(a+1)} \; |a+1\rangle, \qquad S^{-}|a\rangle = \sqrt{S(S+1) - (a-1)a} \; |a-1\rangle,
\ee
with the understanding that $S^{+} |S\rangle = S^{-} |-S\rangle = 0$. Then let $S^{1} = \frac12 (S^{+}+S^{-})$, $S^{2} = \frac1{2\ii} (S^{+}-S^{-})$, and $S^{3} |a\rangle = a |a\rangle$. It is well-known that these operators satisfy the spin commutation relations. Further, the matrix elements of $S^{1}, S^{\pm}$ are all nonnegative, and the matrix elements of $S^{2}$ are all less than or equal to those of $S^{1}$ in absolute values. Using the Trotter formula and multiple resolutions of the identity, we have
\be
\label{Trotter estimate}
\begin{split}
\bigl| \Tr S_{0}^{2} S_{x}^{2} \e{-\beta H_{\Lambda}} \bigr| &\leq \lim_{N\to\infty} \sum_{\sigma_{0}, \dots, \sigma_{N} \in \{-S,\dots,S\}^{\Lambda}} \biggl| \langle \sigma_{0} | S_{0}^{2} S_{x}^{2} | \sigma_{1} \rangle \\
&\langle \sigma_{1} | \e{\frac\beta N \sum J_{yz}^{3} S_{y}^{3} S_{z}^{3}} | \sigma_{1} \rangle \langle \sigma_{1} | \Bigl( 1 + \frac\beta N \sum_{y,z \in \Lambda} ( J_{yz}^{1} S_{y}^{1} S_{z}^{1} + J_{yz}^{2} S_{y}^{2} S_{z}^{2} ) \Bigr) | \sigma_{2} \rangle \\
\dots &\langle \sigma_{N} | \e{\frac\beta N \sum J_{yz}^{3} S_{y}^{3} S_{z}^{3}} | \sigma_{N} \rangle \langle \sigma_{N} | \Bigl( 1 + \frac\beta N \sum_{y,z \in \Lambda} ( J_{yz}^{1} S_{y}^{1} S_{z}^{1} + J_{yz}^{2} S_{y}^{2} S_{z}^{2} ) \Bigr) | \sigma_{0} \rangle \biggr|.
\end{split}
\ee
Observe that the matrix elements of all operators are nonnegative, except for $S_{0}^{2} S_{x}^{2}$. Indeed, this follows from
\be
J_{yz}^{1} S_{y}^{1} S_{z}^{1} + J_{yz}^{2} S_{y}^{2} S_{z}^{2} = \tfrac14 (J_{yz}^{1} - J_{yz}^{2}) (S_{y}^{+} S_{z}^{+} + S_{y}^{-} S_{z}^{-}) + \tfrac14 (J_{yz}^{1} + J_{yz}^{2}) (S_{y}^{+} S_{z}^{-} + S_{y}^{-} S_{z}^{+}) .
\ee
We get an upper bound for the right side of \eqref{Trotter estimate} by replacing $|\langle \sigma_{0}| S_{0}^{2} S_{x}^{2} |\sigma_{1}\rangle|$ with $\langle\sigma_{0}| S_{0}^{1} S_{x}^{1} |\sigma_{1}\rangle$. We have obtained
\be
\bigl| \Tr S_{0}^{2} S_{x}^{2} \e{-\beta H_{\Lambda}} \bigr| \leq \Tr S_{0}^{1} S_{x}^{1} \e{-\beta H_{\Lambda}},
\ee
which proves the first claim. The second claim can be proved exactly the same way.
\end{proof}

\begin{corollary}
Assume that for all $x,y \in \Lambda$, the couplings satisfy
\[
J_{xy}^{1} = J_{xy}^{2} \geq 0.
\]
Then we have for all $x,y,z,u \in \Lambda$
\[
\frac\partial{\partial J_{xy}^{1}} \langle S_{z}^{2} S_{u}^{2} \rangle \leq \frac\partial{\partial J_{xy}^{1}} \langle S_{z}^{1} S_{u}^{1} \rangle.
\]
\end{corollary}

\begin{proof}
For $i=1,2,3$, we have
\be
\frac1\beta \frac\partial{\partial J_{xy}^{1}} \langle S_{z}^{i} S_{u}^{i} \rangle  = (S_{x}^{1} S_{y}^{1}, S_{z}^{i} S_{u}^{i}) - \langle S_{x}^{1} S_{y}^{1} \rangle \langle S_{z}^{i} S_{u}^{i} \rangle,
\ee
where $(A,B)$ denotes the Duhamel two-point function,
\be
(A,B) = \frac1{Z_{\Lambda}} \int_{0}^{1} \Tr A \e{-s\beta H_{\Lambda}} B \e{-(1-s) \beta H_{\Lambda}} \dd s.
\ee
It is not hard to extend the proof of Theorem \ref{thm quantum correl ineq} to the Duhamel function, so that
\be
\bigl| (S_{x}^{1} S_{y}^{1}, S_{z}^{2} S_{u}^{2}) \bigr| \leq (S_{x}^{1} S_{y}^{1}, S_{z}^{1} S_{u}^{1}).
\ee
Further, we have $\langle S_{z}^{2} S_{u}^{2} \rangle = \langle S_{z}^{1} S_{u}^{1} \rangle$ by symmetry. The result follows.
\end{proof}

\section{Decay of correlations due to symmetries}
\label{sec MW}

In this section we prove a variant of the Mermin-Wagner theorem. Our method of proof only works for systems that are effectively two-dimensional. The first result, with an explicit bound on the two-point correlation function, is due to Fisher and Jasnow \cite{FJ}. Unfortunately, it only yields logarithmic decay. The decay is, however, expected to be power-law, and this was proven by McBryan and Spencer \cite{MS} in a short and lucid article that exploits complex rotations. Power-law decay was proven for some quantum systems in \cite{BFK,Ito}. The proofs use Fourier transform and the Bogolubov inequality, and they are limited to regular two-dimensional lattices. A much more general result was obtained by Koma and Tasaki using complex rotations \cite{KT}. The present proof is similar to theirs but somewhat simpler.
Absence of ordering and of symmetry breaking was proven in \cite{FP1,FP2}.

We assume that $J_{xy}^{1} = J_{xy}^{2}$ for all $x,y$.
The decay of correlations is measured by the following expression:
\be
\xi_{\beta}(x) = \suptwo{(\phi_{y}) \in \bbR^{\Lambda}}{\phi_{x}=0} \biggl[ \phi_{0} - 2\beta S^{2} \sum_{y,z \in \Lambda} |J_{yz}^{1}| \bigl( \cosh(\phi_{y}-\phi_{z}) - 1 \bigr) \biggr].
\ee
The solution of this variational problem is essentially a discrete harmonic function. We can estimate it explicitly in the case of ``2D-like'' graphs with nearest-neighbor couplings. Let $\Lambda$ denote a graph, i.e.\ a finite set of vertices and a set of edges, and let $d(x,y)$ denote the graph distance, i.e.\ the length of the shortest path that connects $x$ and $y$.

\begin{lemma}
\label{lem xi}
Assume that $J_{xy}^{i}=0$ whenever $d(x,y)\neq1$ and let $J = \max|J_{xy}^{i}|$. Assume in addition that there exists a constant $K$ such that, for any $\ell\geq0$,
\[
\#\bigl\{ \{x,y\} \subset \Lambda : d(0,x) = \ell, d(0,y) = \ell+1, \text{ and } d(x,y)=1 \bigr\} \leq K(\ell+1).
\]
Then there exists $C = C(\beta,S,J,K)$, which does not depend on $x$, such that
\[
\xi_{\beta}(x) \geq \frac1{16\beta J S^{2} K} \log\bigl( d(0,x)+1 \bigr) - C.
\]
\end{lemma}

\begin{proof}
With $c$ to be chosen later, let
\be
\phi_{y} = \begin{cases} c \log\frac{d(0,x)+1}{d(0,y)+1} & \text{if } d(0,y) \leq d(0,x), \\ 0 & \text{otherwise.} \end{cases}
\ee
Then
\be
\xi_{\beta}(x) \geq c \log(d(0,x)+1) - 4\beta S^{2} JK \sum_{\ell=0}^{d(0,x)-1} \bigl( \cosh ( c \log\tfrac{\ell+2}{\ell+1} ) - 1 \bigr) (\ell+1).
\ee
From Taylor expansions of the logarithm and of the hyperbolic cosine, there exist $C,C'$ such that
\be
\begin{split}
\xi_{\beta}(x) &\geq c \log(d(0,x)+1) - 4\beta S^{2} JK c^{2} \sum_{\ell=1}^{d(0,x)} \frac1\ell - C' \\
&\geq \bigl[ c - 4\beta S^{2} JK c^{2} \bigr] \log(d(0,x)+1) - C.
\end{split}
\ee
The optimal choice is $c = (8\beta S^{2} JK)^{-1}$.
\end{proof}

\begin{theorem}
Assume that $J_{xy}^{1} = J_{xy}^{2}$ for all $x,y \in \Lambda$. Then, for $i =1,2$, we have
\[
\bigl| \langle S_{0}^{i} S_{x}^{i} \rangle \bigr| \leq S^{2} \e{-\xi_{\beta}(x)}.
\]
\end{theorem}

In the case of 2D-like graphs, we can use Lemma \ref{lem xi} and we obtain algebraic decay with a power greater than $(8\beta JS^{2}K)^{-1}$.

\begin{proof}
We use the method of complex rotations. Let
\be
S_{y}^{\pm} = S_{y}^{1} \pm \ii S_{y}^{2}.
\ee
One can check that for any $a \in \bbC$, we have
\be
\e{a S_{y}^{3}} S_{y}^{\pm} \e{-a S_{y}^{3}} = \e{\pm a} S_{y}^{\pm}.
\ee
We have $\langle S_{0}^{+} S_{x}^{-} \rangle = 2 \langle S_{0}^{1} S_{x}^{1} \rangle$ and this is nonnegative by Theorem \ref{thm quantum correl ineq}.
The Hamiltonian \eqref{def spin Ham} can be rewritten as
\be
H_{\Lambda} = -\tfrac12 \sum_{y,z \in \Lambda} \bigl( J_{yz}^{1} S_{y}^{+} S_{z}^{-} + J_{yz}^{3} S_{y}^{3} S_{z}^{3} \bigr)
\ee
Given numbers $\phi_{y}$, let
\be
A = \prod_{y\in\Lambda} \e{\phi_{y} S_{y}^{3}}.
\ee
Then
\be
\Tr S_{0}^{+} S_{x}^{-} \e{-\beta H_{\Lambda}} = \Tr A S_{0}^{+} S_{x}^{-} A^{-1} \e{-\beta A H_{\Lambda} A^{-1}}.
\ee
We now compute the rotated Hamiltonian.
\be
\begin{split}
A H_{\Lambda} A^{-1} &= -\tfrac12 \sum_{y,z \in \Lambda} \bigl( J_{yz}^{1} \e{\phi_{y}-\phi_{z}} S_{y}^{+} S_{z}^{-} + J_{yz}^{3} S_{y}^{3} S_{z}^{3} \bigr) \\
&= H_{\Lambda} -\tfrac12 \sum_{y,z \in \Lambda} J_{yz}^{1} \bigl( \cosh(\phi_{y}-\phi_{z}) - 1 \bigr) \; S_{y}^{+} S_{z}^{-} -\tfrac12 \sum_{y,z \in \Lambda} J_{yz}^{1} \sinh(\phi_{y}-\phi_{z}) \; S_{y}^{+} S_{z}^{-} \\
&\equiv H_{\Lambda} + B + C.
\end{split}
\ee
Notice that $B^{*}=B$ and $C^{*}=-C$. We obtain
\be
\Tr S_{0}^{+} S_{x}^{-} \e{-\beta H_{\Lambda}} = \e{\phi_{0} - \phi_{x}} \Tr S_{0}^{+} S_{x}^{-} \e{-\beta H_{\Lambda} - \beta B - \beta C}.
\ee
We now estimate the trace in the right side using the Trotter product formula and the H\"older inequality for traces. Recall that $\|B\|_{s} = (\Tr |B|^{s})^{1/s}$, with $\|B\|_{\infty} = \|B\|$ being the usual operator norm.
\be
\begin{split}
\Tr S_{0}^{+} S_{x}^{-} \e{-\beta H_{\Lambda} - \beta B - \beta C} &= \lim_{N\to\infty} \Tr S_{0}^{+} S_{x}^{-} \Bigl( \e{-\frac\beta N H_{\Lambda}} \e{-\frac\beta N B} \e{-\frac\beta N C} \Bigr)^{N} \\
&\leq \lim_{N\to\infty} \| S_{0}^{+} S_{x}^{-} \|_{\infty} \; \bigl\| \e{-\frac\beta N H_{\Lambda}} \bigr\|_{N}^{N} \; \bigl\| \e{-\frac\beta N B} \bigr\|_{\infty}^{N} \; \bigl\| \e{-\frac\beta N C} \bigr\|_{\infty}^{N}.
\end{split}
\ee
Observe now that $\|S_{0}^{+} S_{x}^{-}\| = 2 S^{2}$, $\| \e{-\frac\beta N H_{\Lambda}} \|_{N}^{N} = Z_{\Lambda}$, $\| \e{-\frac\beta N B} \|^{N} \leq \e{\beta \|B\|}$, and $\| \e{-\frac\beta N C} \| = 1$. The theorem then follows from
\be
\|B\| \leq S^{2} \sum_{y,z \in \Lambda} |J_{yz}^{1}| \bigl( \cosh(\phi_{y}-\phi_{z}) - 1 \bigr).
\ee
\end{proof}

\medskip
\noindent
{\it Acknowledgments:} We are grateful to the Institute for Advanced Study at Princeton for hospitality. J.F.\ thanks Thomas Spencer for his generosity and interest and the ``Fund for Math'' and ``The Robert and Luisa Fernholz Visiting Professorship Fund'' for financial support.


\begin{thebibliography}{99}

\bibitem{BFK}
C.A.~Bonato, J.~Fernando Perez, A.~Klein,
{\em The Mermin-Wagner phenomenon and cluster properties of one- and two-dimensional systems},
J. Stat. Phys. 29, 159--175 (1982)

\bibitem{BKUel}
C. Borgs, R. Koteck\'y, D. Ueltschi,
{\em Low temperature phase diagrams for quantum perturbations of classical spin systems},
Commun. Math. Phys. 181, 409--446 (1996)

\bibitem{BR}
O.~Bratteli, D.W.~Robinson,
{\em Operator Algebras and Quantum Statistical Mechanics 2},
Springer (1981)

\bibitem{DFF}
N. ~Datta, R. ~Fern\'andez, J. ~Fr\"ohlich, 
{\em Low-temperature phase diagrams of quantum lattice
systems, I. Stability for quantum perturbations of classical systems with finitely many
ground states}, 
J. Stat. Phys. 84, 455--534 (1996)

\bibitem{Dir}
S.~Dirren,
diploma thesis, ETH Z\"urich

\bibitem{DLS} 
F. ~Dyson, E.H. ~Lieb, B. ~Simon, 
{\em Phase transitions in quantum spin systems with isotropic and non isotropic interactions}, 
J. Stat. Phys. 18, 335--383, (1978)

 \bibitem{FJ}
M.E.~Fisher, D.~Jasnow,
{\em Decay of order in isotropic systems of restricted dimensionality. II. Spin systems},
Phys. Rev. B 3, 907--924 (1971)

\bibitem{FILS}
J. ~Fr\"ohlich, R. ~Israel, E.H. ~Lieb, B. ~Simon,
{\em Phase transitions and reflection positivity, I.
General theory and long range lattice models},
Commun. Math. Phys. 62, 1--34 (1978)

\bibitem{FP1}
J.~Fr\"ohlich, C.-\'E.~Pfister,
{\em On the absence of spontaneous symmetry breaking and of crystalline ordering in two-dimensional systems},
Comm. Math. Phys. 81, 277--298 (1981)

\bibitem{FP2}
J.~Fr\"ohlich, C.-\'E.~Pfister,
{\em Absence of crystalline ordering in two dimensions},
Comm. Math. Phys. 104, 697--700 (1986)

\bibitem{FR}
J.~Fr\"ohlich, P.-F.~Rodriguez,
{\em Some applications of the Lee-Yang theorem},
J. Math. Phys. 53, 095218 (2012)

\bibitem{GK}
C.~Gruber, H.~Kunz,
{\em General properties of polymer systems},
Comm. Math. Phys. 22, 133--161 (1971)

\bibitem{GRS}
F.~Guerra, L.~Rosen, B.~Simon,
{\em Correlation inequalities and the mass gap in $P(\phi)_{2}$},
Commun. Math. Phys. 41, 19--32 (1975)

\bibitem{Ito}
K.R.~Ito,
{\em Clustering in low-dimensional $SO(N)$-invariant statistical models with long-range interactions},
J. Stat. Phys. 29, 747--760 (1982)

\bibitem{KT}
T.~Koma, H.~Tasaki,
{\em Decay of superconducting and magnetic correlations in one-and two-dimensional Hubbard models},
Phys. Rev. Lett. 68, 3248--3251 (1992)

\bibitem{KP}
R.~Koteck\'y, D.~Preiss,
{\em Cluster expansion for abstract polymer models},
Commun. Math. Phys. 103, 491--498 (1986)

\bibitem{Kubo}
K.~Kubo,
{\em Existence of long-range order in the XY model},
Phys. Rev. Lett. 61 110--112 (1988).

\bibitem{LP}
J.L.~Lebowitz, O.~Penrose,
{\em On the exponential decay of correlations},
Commun. Math. Phys. 39, 165--184 (1974)

\bibitem{MS}
O.A.~McBryan, T.~Spencer,
{\em On the decay of correlations in $SO(n)$-symmetric ferromagnets},
Commun. Math. Phys. 53, 299--302 (1977)

\bibitem{Pfi}
C.-\'E.~Pfister,
{\em Analyticity properties of the correlation functions for the anisotropic Heisenberg model},
Commun. Math. Phys. 41, 109--117 (1975)

\bibitem{Sim}
B.~Simon,
{\em The Statistical Mechanics of Lattice Gases},
Princeton University Press (1993)

\bibitem{Uel}
D.~Ueltschi,
{\em Cluster expansions and correlation functions},
Moscow Math. J. 4, 511--522 (2004)

\bibitem{Uel2}
D.~Ueltschi,
{\em Random loop representations for quantum spin systems},
J. Math. Phys. 54, 083301 (2013)

\end{thebibliography}
\end{document}